\documentclass[11pt]{article}

\usepackage[T1]{fontenc}
\usepackage[a4paper,margin=1in]{geometry}
\usepackage{amsmath,amssymb,amsthm}
\usepackage{booktabs}
\usepackage{array}
\usepackage{enumitem}
\usepackage[section]{placeins}
\usepackage[expansion=false]{microtype}
\usepackage{xcolor}
\usepackage{tikz}
\usetikzlibrary{arrows.meta,positioning}
\usepackage{hyperref}
\usepackage[authoryear,round]{natbib}

\newcolumntype{P}[1]{>{\raggedright\arraybackslash}p{#1}}

\hypersetup{
  colorlinks=true,
  linkcolor=blue!55!black,
  citecolor=blue!55!black,
  urlcolor=blue!55!black,
    pdftitle={Cost Sharing with Hidden Time Flexibility},
    pdfauthor={Mohsen Pourpouneh and Farzaneh Rajabighamchi},
    pdfsubject={Cost sharing with partially verifiable service-time flexibility},
    pdfkeywords={cost sharing, core allocations, partial verification,
                 interval scheduling, coalitional manipulation}
}

\newtheorem{theorem}{Theorem}[section]
\newtheorem{proposition}[theorem]{Proposition}
\newtheorem{lemma}[theorem]{Lemma}

\theoremstyle{definition}

\newtheorem{example}[theorem]{Example}
\theoremstyle{remark}

\newcommand{\R}{\mathbb{R}}

\title{Cost Sharing with Hidden Time Flexibility}
\author{Mohsen Pourpouneh, Farzaneh Rajabighamchi%
\thanks{Department of Data Analytics and Digitalisation,
School of Business and Economics, Maastricht University,
Maastricht, the Netherlands.
E-mail addresses:
\href{mailto:m.pourpouneh@maastrichtuniversity.nl}
{m.pourpouneh@maastrichtuniversity.nl} (M. Pourpouneh) and
\href{mailto:f.rajabighamchi@maastrichtuniversity.nl}
{f.rajabighamchi@maastrichtuniversity.nl} (F. Rajabighamchi).\\
We thank Burak Can, Janos Flesch, Alexander Grigoriev and Marc Schroder for their comments during the development of this work. }}

\begin{document}
\maketitle

\begin{abstract}
Customers often specify acceptable time intervals for receiving a service. When service on a date entails a fixed activation cost, overlapping intervals allow joint service and cost sharing. Customers may nevertheless conceal flexibility by reporting a smaller interval, affecting both operating cost and its allocation. We ask whether an operator can select a core allocation at every reported profile without making concealment profitable.

We establish sharp population boundaries for preventing such coalitional reductions. When intervals may be narrowed from either end, a core-selecting rule preventing such reductions exists if and only if there are at most three customers. With four or more customers, impossibility holds even under Pareto contraction-proofness, which excludes deviations making every contracting customer weakly better off and at least one strictly better off. The impossibility persists when each contracting customer conceals only an arbitrarily small fraction of her feasible interval.

The boundaries change when either the earliest or latest acceptable date of every customer is verifiable. With common verified dates, equal sharing is core-selecting and prevents such reductions for any population size. When verified dates may differ, the boundary rises to five: such a rule exists for up to five customers, whereas with six or more no core-selecting rule prevents the contracting customers from reducing their combined payment. Under Pareto contraction-proofness, however, core selection is possible for every population size on either one-sided domain. Finally, the relaxations needed for approximate core selection and coalitional contraction-proofness must grow with population size.
\end{abstract}

\medskip

\section{Introduction}\label{sec:introduction}
Many fixed-cost services can be provided jointly when customers' acceptable time windows overlap. Consider several firms at an industrial park that use a common carrier to send small loads to a port terminal. Each load has a departure window, beginning when it becomes ready and ending at the latest departure compatible with its delivery deadline. Dispatching a truck entails a fixed cost and, when capacity is not binding, loads with overlapping windows can travel together and share that cost. By claiming that its load will be ready later or must leave earlier, a firm can appear less flexible and may alter its share of the cost.

This example captures a broader cost-sharing problem. Reporting a smaller interval is credible,
 whereas reporting dates outside the true interval may result in service at an unacceptable time.  Because reported intervals determine both the number of service dates and how the resulting cost is divided, the operator must select core allocations across reports without rewarding such concealment. Each allocation must recover the minimum activation cost while ensuring that no coalition pays more than its stand-alone cost.

We consider four incentive requirements. Individual contraction-proofness (ICP) prevents a customer from lowering her payment by unilaterally reporting a smaller interval. Coalitional contraction-proofness (CCP) prevents the customers who report smaller intervals from reducing their combined payment. Strong coalitional contraction-proofness (SCCP) requires that no coalition containing all customers who report smaller intervals can reduce its total payment. Such a coalition may also include customers whose reports remain unchanged. Hence, SCCP implies CCP, which in turn implies ICP. We also consider Pareto contraction-proofness, which prevents contractions that weakly lower every contracting customer's payment and strictly lower at least one. CCP implies Pareto contraction-proofness, since every such contraction also reduces the contracting customers' combined payment.

Our first results concern the case in which each customer may report any nonempty subinterval of her true interval. We call this the two-sided domain because a customer may increase the left endpoint, decrease the right endpoint, or do both. On this domain, we show that a rule can be both core-selecting and ICP, regardless of the number of customers. A population bound arises only when coalitional deviations are considered: we construct a core-selecting SCCP rule for up to three customers, whereas no core-selecting rule satisfies even CCP from four customers onward. The impossibility continues to hold when customers conceal only an arbitrarily small fraction of their flexibility. Importantly, the model imposes no directional restriction on a customer: she may contract from the left, from the right, or from both sides. The result also extends to approximate rules: repeating the construction across disjoint four-customer groups shows that the permitted violations of the core and CCP requirements must increase with the number of customers.

The four-customer obstruction relies on deviations that narrow intervals in opposing directions. This leads us to ask what happens when either every left endpoint or every right endpoint is verifiable and cannot be changed. In the freight example, warehouse records may verify ready dates, while contractual deadlines and known travel times may verify latest departure dates. In this domain, we show that when every customer has the same endpoint, e.g., a common terminal cutoff, there exists a cost-sharing rule that is core-selecting and SCCP for any number of customers. When the verified endpoint values may differ, the exact population boundary rises from three to five, in this domain compared to the two-sided domain, and no core-selecting rule satisfies even CCP from six customers onward. The six-customer impossibility also persists under arbitrarily small concealment, and replicating the construction yields population-dependent lower bounds for approximate rules.  We further show that the four-customer impossibility holds under Pareto contraction-proofness. On either one-sided domain, a core-selecting rule satisfying this weaker requirement exists for every number of customers. Thus, although the two-sided population
boundary remains three, verifying one endpoint removes the population restriction under Pareto contraction-proofness.

The remainder of the paper is organized as follows. Section~\ref{sec:literature} discusses the related literature, Section~\ref{sec:model} presents the model and incentive requirements. Sections~\ref{sec:boundary} and~\ref{sec:onesided} establish the exact population boundaries, Section~\ref{sec:robustness} presents the robustness and approximate results. Finally, Section~\ref{sec:conclusion} discusses Pareto contraction-proofness, the Shapley value and the nucleolus, and concludes the paper.

\section{Related Literature}\label{sec:literature}
Our analysis is related to the literature on cooperative minimum-coloring games. For a fixed interval profile, the service problem can be represented as a coloring problem in which customers with disjoint intervals are incompatible. Core allocations and cost-sharing rules for minimum-coloring and related combinatorial optimization games have been studied extensively \citep{deng1999,bietenhader2006,okamoto2008,hamers2014,baheltrudeau2022}. This literature treats the incompatibility graph, and hence the cooperative cost game, as fixed. In our setting, reported intervals determine the graph, which may change when customers conceal flexibility. In contrast, we consider whether core allocations can be selected consistently across the cost games induced by different reports.

The incentive question connects our work to strategy-proof cost sharing. \citet{moulin1999} characterizes incremental cost-sharing mechanisms using coalition strategy-proofness, while \citet{moulinshenker2001} study strategy-proof cost sharing for submodular costs and the tension between efficiency and budget balance. This literature generally allows agents to report valuations, demands, or participation decisions, and the mechanism determines both who is served and how much they pay. In contrast, all customers in our model must be served; their reports concern when they can be served, not whether they receive service. These reports can alter both the minimum operating cost and the cooperative game from which payments are selected.

Our reporting restriction is a form of partial verification: a customer may hide feasible service dates but cannot credibly report dates outside her true interval. \citet{greenlaffont1986} model such restrictions through type-dependent message spaces, while subsequent work examines how verifiability and evidence affect implementation \citep{singhwittman2001,kartiktercieux2012}. In our setting, admissible reports are ordered by interval inclusion: a customer may increase the left endpoint, decrease the right endpoint, or do both. When one endpoint is verifiable, only one of these directions remains available. \citet{krahmerstrausz2025} similarly study how restricting misreporting to one direction affects incentive compatibility. We examine how this distinction affects core selection when customers may coordinate their reports.

Several papers study strategic behavior in interval and scheduling environments. \citet*{deligkas2024} consider truthful interval covering without transfers. \citet{tamir2023} studies a noncooperative real-time scheduling game in which jobs with release dates and deadlines share server activation costs under a given proportional rule. \citet*{vijayalakshmi2026} analyze equilibrium and inefficiency when players strategically place jobs in an interval-scheduling game. \citet{munich2024} studies cooperative schedule games for access to a non-rival resource, including core nonemptiness and equal-pooling allocations. These models share either the interval structure or the fixed activation cost of our setting. Our setting differs in that all customers must be served and the payment rule must select core allocations consistently across reported profiles while preventing profitable concealment of flexibility.

Time windows are also central to vehicle routing and service operations. Classical work develops routing and scheduling methods for exogenously given time windows \citep{solomon1987,solomondesrosiers1988}, while later studies examine time-slot design, pricing, and customer incentives in attended delivery and service systems \citep{agatz2011,klein2019,vinsensius2020,ulmer2024,visser2024}. This literature typically treats time windows as operational constraints or choices from a menu and focuses on routing costs, demand management, or service quality. Cooperative-game methods have also been used to allocate the gains from transportation collaboration \citep{frisk2010,guajardoronnqvist2016}. Our model does not include routing or capacity constraints. Instead, reported time windows determine the number of fixed-cost activations required to serve each coalition. We ask when core allocations can be selected across reports without making it profitable for customers, acting alone or together, to hide flexibility.

\section{Model and Notation}\label{sec:model}

Let $N=\{1,\ldots,n\}$ be the set of customers, and let $\mathcal{I}_n$ be the set of all profiles of nonempty compact real intervals. Let \(I=(I_i)_{i\in N}\in\mathcal{I}_n\) denote the profile of customers' feasible intervals. That is, customer $i$ can be served on any date in the interval $I_i=[\ell_i,r_i]\subseteq\R$ with $\ell_i\leq r_i$. Let $J=(J_i)_{i\in N}\in\mathcal{I}_n$ denote the submitted profile satisfying $J_i\subseteq I_i$ for every $i\in N$. That is, the customers are assumed not to report dates outside their true feasible intervals, since doing so could result in service on a date that does not satisfy their requirements. Therefore, we restrict attention to reports that conceal feasible dates. Let $F>0$ denote the service activation cost on each date. We assume that the activation cost is independent of the number of customers served on that date.  For a group of customers $S\subseteq N$, let $\tau_I(S)$ denote the minimum number of dates needed to serve all customers in $S$.  Formally, 
\begin{equation*}
\tau_I(S)=
\min\bigl\{|T|:T\subseteq\R\text{ finite and }
T\cap I_i\neq\varnothing \text{ for every }i\in S\bigr\}.   
\end{equation*}

The associated cooperative cost of a coalition $S$ is denoted by $c_I(S)=F \tau_I(S)$, and $c_I(\varnothing)=0$.
A cost-sharing rule is a mapping $x:\mathcal{I}_n \rightarrow \mathbb{R}^N$, where $x(I) = (x_i(I))_{i\in N}$.  A cost-sharing rule $x$ is \emph{core-selecting} whenever $\sum\limits_{i\in N}x_i(I) = c_I(N)$ and $\sum\limits_{i\in S}x_i(I) \leq c_I(S)$ for every $S\subseteq N$. This ensures that the collected payments equal the
operator's cost and no group pays more than it would cost that group to arrange service on its own.


We denote the set of customers who submit a strictly smaller interval by
$D(I,J)=\{i\in N:J_i\subsetneq I_i\}$. We refer to the members of
$D(I,J)$ as the \emph{contracting customers}. We distinguish the following
contraction domains:
\begin{itemize}
    \item \emph{Two-sided} if
    $J_i=[\ell_i',r_i']$ with
    $\ell_i\leq\ell_i'\leq r_i'\leq r_i$ for every $i\in N$.
    A contracting customer may change the left endpoint, the right endpoint,
    or both.
    \item \emph{Left-only} if $J_i=[\ell_i',r_i]$ with
    $\ell_i\leq\ell_i'\leq r_i$ for every $i\in N$. Here the right
    endpoint is verified.
    \item \emph{Right-only} if $J_i=[\ell_i,r_i']$ with
    $\ell_i\leq r_i'\leq r_i$ for every $i\in N$. Here the left
    endpoint is verified.
\end{itemize}


A cost-sharing rule \(x\) is \emph{individually contraction-proof}
(ICP) if a customer cannot reduce her payment by unilaterally
contracting her interval. Formally, for every profile
\(I\in\mathcal{I}_n\), every admissible submitted profile
\(J\in\mathcal{I}_n\), and every customer \(i\in N\) such that
\(D(I,J)=\{i\}\),
\begin{equation}
x_i(J)\geq x_i(I).
\label{eq:icp}
\end{equation}




The next condition, rules out a reduction in the coalition's total payment, even when coalition members can redistribute
gains among themselves. In particular, it rules out any deviation that weakly lowers every
coalition member's payment and strictly lowers at least one member's payment. Formally, a cost-sharing rule $x$ is \emph{strongly coalitionally contraction-proof} (SCCP) if for every profile $I$, every admissible submitted profile $J$, and every coalition \(C\subseteq N\) satisfying \(D(I,J)\subseteq C\),
\begin{equation}
\sum_{i\in C}x_i(J)\geq\sum_{i\in C}x_i(I).
\label{eq:sccp}
\end{equation}

A weaker requirement is obtained by restricting attention to coalitions consisting only of contracting customers. Under this interpretation, every member of the manipulating coalition must report a strictly smaller interval. Therefore the customers that do not misreport are not included in the coalition (as well as share in its payment gain). Formally, a cost-sharing rule $x$ is \emph{coalitionally contraction-proof} (CCP) whenever 
\begin{equation}
\sum_{i\in D(I,J)}x_i(J)
\geq
\sum_{i\in D(I,J)}x_i(I).
\label{eq:CCP}
\end{equation}
for every admissible pair $I,J$ with $D(I,J)\neq\varnothing$. 




\begin{example}\label{ex:calibration}
Five firms at an industrial park use a common carrier to send small loads to a port terminal.  Let $F=100$ be the cost of each truck dispatch. Suppose their acceptable departure intervals are $I=\bigl([0,2], [1,4], [3,7], [6,9], [8,11]\bigr)$. It is easy to verify that $\tau_I(N)=3$, and that the  unique core allocation is $x(I)=(100,0,100,0,100)$. Now let firms $\{1, 3, 5\}$ to report $J_1 = [0,0.5] \subset I_1$, $J_3 = [3.5, 6.5] \subset I_3$, and $J_5 = [10, 11] \subset I_5$. That is, $D(I, J) =\{1, 3, 5\}$. After these misreport, the number of required dispatches increases to four. The corresponding core allocation is again unique and is given by $x(J) = (100, 100, 0, 100, 100)$. Although the total service cost increases from $300$ to $400$, the combined payment of the contracting customers decreases. Before the contraction, the contracting customers pay $\sum\limits_{i\in D(I,J)} x_i(I)  = 300$. After the contraction, they pay $\sum_{i\in D(I,J)} x_i(J) = 200$.  Therefore, $\sum_{i\in D(I,J)} x_i(J) < \sum_{i\in D(I,J)} x_i(I)$.  This inequality violates CCP. Since SCCP implies CCP, the same deviation also violates SCCP. 
\end{example}

\section{Two-Sided Domain}
\label{sec:boundary}

This section considers the two-sided domain, in which every customer may independently contract from the left, from the right, or from both sides, and examines the compatibility of core selection with ICP, SCCP, and CCP. We first establish core selection and ICP for every population size, then present the possibility and impossibility results and derive the exact population boundaries.

\subsection{Core selection and ICP for all \texorpdfstring{$n$}{n}}
\label{subsec:greedy-icp}

We first show that a core-selecting ICP rule exists for every number of
customers. Consider the following right-endpoint greedy procedure.

Starting with all customers unserved, select a customer with the
smallest right endpoint and activate service at that endpoint. All
customers whose intervals contain this date are then served and
removed. Repeat this procedure until every customer is served. Ties
are broken by the smallest customer index.

Let \(P(I)\) be the set of customers selected by this procedure. Define
the \emph{greedy rule} by
\begin{equation}
x_i^{\mathrm G}(I)
=
\begin{cases}
F, & i\in P(I),\\
0, & i\notin P(I).
\end{cases}
\label{eq:greedy-rule}
\end{equation}
Thus, every selected customer pays the cost of one activation, while
all other customers pay zero.

\begin{theorem}[Core selection and ICP]
\label{thm:greedy-icp}
The greedy rule is core-selecting and ICP on the
two-sided domain.
\end{theorem}

\begin{proof}
The intervals of the customers in \(P(I)\) are pairwise disjoint.
Indeed, after an activation at the smallest right endpoint, every
remaining interval must begin strictly after that endpoint. Moreover, the dates selected by the procedure serve all customers. Therefore, the number of selected customers is exactly the minimum number of required dates, that is, $|P(I)|=\tau_I(N)$. Therefore, $\sum\limits_{i\in N}x_i^{\mathrm G}(I) = F|P(I)| = F\tau_I(N) = c_I(N)$. For every coalition $S\subseteq N$, the customers in $P(I)\cap S$ have pairwise-disjoint intervals. Hence, each of these customers requires a distinct service date. Since all of them belong to $S$, serving coalition $S$ requires at least $|P(I)\cap S|$ dates. Therefore, $|P(I)\cap S|\leq\tau_I(S)$, which implies $\sum\limits_{i\in S}x_i^{\mathrm G}(I) = F|P(I)\cap S| \leq F\tau_I(S) = c_I(S)$.
Thus, the greedy rule is core-selecting.

Next, we show that the greedy rule satisfies ICP. Consider a unilateral
contraction from $I$ to $J$ with $D(I,J)=\{k\}$. If customer $k$ pays zero
at $I$, her payment cannot decrease. If customer $k$ pays $F$, then she is selected at $I$. Every date chosen before her selection lies
outside $I_k$, and hence outside $J_k$. Since all other intervals are
unchanged  and her right endpoint cannot increase, customer $k$ is selected at $J$ at the same step as at $I$ or earlier. Thus, customer $k$ is also selected at $J$, and must pay $F$, which completes the proof.
\end{proof}

Note that,  core selection and ICP do not uniquely determine the greedy rule. To see this, consider the rule $x^{\mathrm L}$ that repeatedly selects an unserved customer with the largest left endpoint and activates
service at that date. All customers whose intervals contain the
selected date are then served and removed. As before, selected
customers pay $F$, all other customers pay zero, and ties are broken
by the smallest customer index. This rule is obtained by reflecting
the time axis in the right-endpoint greedy procedure. Hence,
Theorem~\ref{thm:greedy-icp} implies that it is also core-selecting
and ICP.

Furthermore, for any fixed $\lambda\in[0,1]$, define
\[
x^\lambda(I)=\lambda x^{\mathrm G}(I)+(1-\lambda)x^{\mathrm L}(I).
\]
The rule $x^\lambda$ is core-selecting because a weighted average
of two core allocations remains in the core. It also satisfies ICP,
since averaging the corresponding ICP inequalities preserves the
inequality. 

We next examine whether core selection remains compatible with the stronger coalitional requirements of SCCP and CCP.

\subsection{Population boundaries}

\subsubsection{Possibility for \texorpdfstring{$n\leq3$}{n <= 3}}


The minimum number of dates required to serve a coalition can equivalently be determined from its mutually incompatible customers. In particular, $\tau_I(S)$ is the largest number of customers in $S$ whose feasible intervals are pairwise disjoint. Indeed, customers with pairwise-disjoint intervals require different service dates, and, for intervals, this lower bound can always be attained. 


\paragraph{The geometry-aware rule for three customers.}


For three customers, the geometric rule $x^{geo}$ is defined as follows:
\begin{enumerate}
\item If all three customers can be served on the same date, the activation cost is divided equally among them.

\item If exactly one pair of customers cannot be served together, each customer in that pair pays \(F\), while the remaining customer pays nothing.

\item If one customer cannot be served with either of the other two, while those two can be served together, the first customer pays \(F\). The remaining \(F\) is charged to whichever of the other two customers is farther from that customer's interval. If their distances are equal, they divide the remaining cost equally.

\item If no two customers can be served together, every customer pays \(F\).
\end{enumerate}

\paragraph{Monotonicity and the three-customer result.}

The next two lemmas isolate the monotonicity needed for SCCP. Their proofs use only the four cases above.

\begin{lemma}[Own-share monotonicity]\label{lem:own}
Let $n=3$, and suppose that only customer $i$ contracts its interval. Then
\[
x_i^{\mathrm{geo}}(J)\geq x_i^{\mathrm{geo}}(I).
\]
\end{lemma}
\begin{proof}
Under $x^{\mathrm{geo}}$, the payment of customer $i$ is one of the values of $\{0,F/3,F/2, F\}$. If customer $i$ pays $F/3$ at profile $I$, the three reported intervals have a common point. If they still have a common point after the contraction, then customer $i$ will still pay $F/3$. Otherwise, $J_i$ is disjoint from at least one other interval, and customer $i$ must pay $F$.

Suppose customer $i$ pays $F/2$. Then one of the other customers, say $k$, is disjoint from both $i$ and $j$, while $I_i$ and $I_j$ intersect and are equally far from $I_k$. If $J_i\cap I_j=\varnothing$, the three intervals are pairwise disjoint, so customer $i$ pays $F$. Otherwise, contracting $I_i$ cannot move it closer to $I_k$, while $I_j$ and $I_k$ remain unchanged. Thus, $J_i$ is at least as far from $I_k$ as $I_j$, and customer $i$ pays either $F/2$ or $F$.

Finally, suppose customer $i$ pays $F$. If customer $i$ belongs to the unique disjoint pair, or is disjoint from both other intervals, or the three intervals are pairwise disjoint, then $J_i\subseteq I_i$ implies that customer $i$ must pay $F$. Otherwise, customer $i$ is the farther of two intersecting intervals from an interval disjoint from both. Contracting $I_i$ cannot reduce this distance; if the two intervals become disjoint, all three intervals are pairwise disjoint. Thus, customer $i$ must pay $F$. Therefore, $x_i^{\mathrm{geo}}(J)\geq x_i^{\mathrm{geo}}(I)$. 
\end{proof}

\begin{lemma}[Fixed-outsider monotonicity]\label{lem:outsider}
Let $n=3$. Suppose customer $i$ leaves its interval unchanged, while the other two customers may contract their intervals. If $\tau_J(N)=\tau_I(N)$, then 
\[
x_i^{\mathrm{geo}}(J)\leq x_i^{\mathrm{geo}}(I).
\]
\end{lemma}

\begin{proof}
Let \(j\) and \(k\) denote the other two customers. If $\tau_I(N) = \tau_J(N)=1$, customer $i$ pays $F/3$ at both profiles. If $\tau_I(N) = \tau_J(N)=3$, customer \(i\) pays \(F\) at both profiles. Hence, in both cases $x_i^{\mathrm{geo}}(J)\leq x_i^{\mathrm{geo}}(I)$. Let $\tau_I(N)=\tau_J(N)=2$.  
For three intervals, there are two possible cases.
\begin{enumerate}
    \item Exactly one pair of intervals is disjoint. If customer $i$ belongs to the disjoint pair, then $x_i^{\mathrm{geo}}(I)=F$, so the claim is immediate. Now suppose, the 
disjoint pair consists of customers $j$ and $k$, that is $I_j\cap I_k=\varnothing$, and $I_i\cap I_j\neq\varnothing$, and $I_i\cap I_k\neq\varnothing$, with \(x_i^{\mathrm{geo}}(I)=0\).

If \(J_j\cap I_i\neq\varnothing\) and
\(J_k\cap I_i\neq\varnothing\), then \(\{j,k\}\) is the only disjoint pair at \(J\), and customer \(i\) pays zero. The two reported intervals cannot both be disjoint from \(I_i\), because that would imply $\tau_J(N)=3$. Suppose, without loss of generality, that $J_j\cap I_i=\varnothing$. If $I_j$ lies to the left of $I_k$, then $\ell_i\leq r_j<\ell_k\leq\ell_k'$, therefore $J_k$ is farther from $J_j$ than $I_i$. If $I_k$ lies to the left of $I_j$, then $r_k'\leq r_k<\ell_j\leq r_i$, and again $J_k$ is farther from $J_j$ than $I_i$. Therefore, customer $i$ pays zero at $J$. 

    \item Exactly two pairs of intervals are disjoint. If $I_i$ is disjoint from both other intervals, then $x_i^{\mathrm{geo}}(I)=F$, so the claim is immediate. Now suppose that
$I_k$ is disjoint from both $I_i$ and $I_j$, while $I_i\cap I_j\neq\varnothing$. Since $\tau_J(N)=2$, we must have $I_i\cap J_j\neq\varnothing$. If $I_k$ lies to the left of $I_i$  and $I_j$, then $J_k$ also lies to their left. The farther of $I_i$ and $J_j$ from $J_k$ is  determined by their left endpoints. Since $\ell_i$ is unchanged and $\ell_j'\geq\ell_j$, customer $i$'s payment cannot increase.  If $I_k$ lies to the right of $I_i$ and $I_j$, then $J_k$ also lies to their right. The farther of $I_i$ and $J_j$ from $J_k$ is determined by their right endpoints. Since $r_i$ is unchanged and $r_j'\leq r_j$, customer $i$'s payment cannot increase.
\end{enumerate}
\end{proof}

\begin{theorem}[Three-customer positive result]\label{thm:three}
Let $n=3$. The rule $x^{\mathrm{geo}}$ is core-selecting and SCCP on the two-sided domain. Therefore, it is also CCP.
\end{theorem}
\begin{proof}
We first show that the rule is core-selecting. Considering the four cases shows that $\sum\limits_{i\in N}x_i^{\mathrm{geo}}(I) =F\tau_I(N)$. Moreover, every customer pays at most $F$. Hence, the core constraint holds for every singleton and every pair of customers with disjoint intervals. If two customers have intersecting intervals, the four cases show that their combined payment is at most $F$, which equals their coalition cost. These are all the proper nonempty coalitions when $n=3$. Therefore, $x^{\mathrm{geo}}(I)$ belongs to the core.

We next show that the rule is SCCP.  Consider an admissible contraction from \(I\) to \(J\) and a coalition \(C\)
satisfying \(D(I,J)\subseteq C\). Since contractions cannot decrease the
number of required service dates, $\tau_J(N)\geq\tau_I(N)$. Based on the size of $C$ the following cases are possible:
\begin{itemize}
    \item $|C| = 1$, the result follows from Lemma~\ref{lem:own}.
    \item $|C| = 2$, let $i$ be the customer outside $C$. If $\tau_J(N)=\tau_I(N)$, the total payment is unchanged, and by Lemma~\ref{lem:outsider} customer $i$'s payment cannot increase.  Therefore, the coalition's total payment cannot decrease. If $\tau_J(N)>\tau_I(N)$, the total payment increases by at least $F$, while customer $i$'s payment can increase by at most $F$. Thus, the coalition's total payment cannot decrease.
    \item $|C| = 3$,  the result follows directly from $\tau_J(N)\geq\tau_I(N)$. 
    \end{itemize}
    Therefore, no coalition containing every contracting customer can reduce its total payment, so the rule is SCCP. Therefore, it is also CCP.
\end{proof}

\subsubsection{Impossibility for \texorpdfstring{$n\geq4$}{n >= 4}}

We now establish the four-customer impossibility. We construct two profiles that can both be contracted to the same reported profile. Core-selecting uniquely determines the payments at the two initial profiles, while leaving one payment parameter at the common reported profile. Applying CCP on the first profile forces this parameter to one extreme,  applying it on the second forces the parameter to the opposite extreme.  The construction uses a parameter $\rho$ to control the amount of concealed flexibility. Fix $0<\rho<1/2$ and consider the following three profiles:
\begin{align}
I_\rho^A 
& = \bigl([0,1-\rho], [1-\frac{\rho}{2},2], [2,3], [3,4] \bigr),\label{eq:obstruction-A}\\
J_\rho
& = \bigl([0,1-\rho],[1,2],[2,3],[3+\rho,4]\bigr),\label{eq:obstruction-J}\\
I_\rho^B
& = \bigl([0,1], [1,2], [2,3+\frac{\rho}{2}],[3+\rho,4]\bigr) \label{eq:obstruction-B}
\end{align}
To get to $J_\rho$ from $I_\rho^A$, customers $2$ and $4$ increase their left endpoints. Therefore, $D(I_\rho^A,J_\rho)=\{2,4\}$. Also, to get to $J_\rho$ from $I_\rho^B$, customers $1$ and $3$ must decrease their right endpoints. Therefore, $D(I_\rho^B,J_\rho)=\{1,3\}.$ Each deviation  involves exactly two contracting customers, and each contracting customer changes only one endpoint. It can be verified that the unique core allocations at $I_\rho^A$ and $I_\rho^B$ are, $x(I_\rho^A)=(F,F,0,F)$  and $x(I_\rho^B)=(F,0,F,F)$. Furthermore, at $J_\rho$, every core allocation is of the form $x(J_\rho) = (F,\alpha F,(1-\alpha)F,F)$, with $0\leq\alpha\leq1$. These core allocations produce the four-customer impossibility.

\begin{theorem}[CCP impossibility]\label{thm:ccp-impossibility}
For every $n\geq4$, no core-selecting CCP rule exists on the  two-sided domain. Therefore, no core-selecting SCCP rule exists on this domain.
\end{theorem}
\begin{proof}
Suppose $x$ is a core-selecting CCP rule. For the contraction \(I_\rho^A\to J_\rho\), we have \(D(I_\rho^A,J_\rho)=\{2,4\}\). Therefore, CCP requires $x_2(J_\rho)+x_4(J_\rho) \geq x_2(I_\rho^A)+x_4(I_\rho^A)$. Using the core allocations above, this becomes $(1+\alpha)F\geq2F$, which forces \(\alpha=1\).

For the contraction \(I_\rho^B\to J_\rho\), we have \(D(I_\rho^B,J_\rho)=\{1,3\}\). Therefore, CCP requires $x_1(J_\rho)+x_3(J_\rho) \geq x_1(I_\rho^B)+x_3(I_\rho^B)$. This becomes $(2-\alpha)F\geq2F$,  which forces $\alpha=0$. This contradicts \(\alpha=1\), proving the result 
for \(n=4\).

For $n>4$, append the same fixed interval to all three profiles for every additional customer $h\in\{5,\ldots,n\}$:
\[I_{\rho,h}^A = I_{\rho,h}^B = J_{\rho,h} = [2h,2h+1].
\]
These additional intervals are pairwise disjoint and are disjoint from the intervals of customers \(1,\ldots,4\). For every resulting profile \(P\), 
\[
\tau_P(N)
=
\tau_P(\{1,2,3,4\})+(n-4).
\]
Moreover, for every \(h\in\{5,\ldots,n\}\), we have $\tau_P(N\setminus\{h\})=\tau_P(N)-1$.
The total-payment condition and the core constraint for
\(N\setminus\{h\}\)  imply $x_h(P)\geq F$.
The singleton core constraint gives \(x_h(P)\leq F\), so every additional
customer pays \(F\). Hence, the payments of customers $\{1,\ldots,4 \}$ sum to \(F\tau_P(\{1,2,3,4\})\) and satisfy all core constraints of the original
four-customer profile. Their payments must  have the same forms
derived above. Since the additional customers do not change their reports,
the two CCP inequalities again force \(\alpha=1\) and \(\alpha=0\). Since SCCP
implies CCP, the impossibility also holds under SCCP. Both deviations are admissible on the two-sided domain; CCP must hold for these deviations even though each deviating customer in the construction contracts only one endpoint.
\end{proof}


Combining the three-customer construction with the four-customer impossibility gives the exact population boundary on the two-sided domain.

\begin{theorem}[Exact population boundary]\label{thm:population-boundary}

On the  two-sided domain, a core-selecting
 SCCP rule exists if and only if \(n\leq3\). The same boundary holds for
core-selecting CCP rules.
\end{theorem}
\begin{proof}
For \(n=1\), charge the sole customer \(F\). For \(n=2\), charge each
customer \(F/2\) when the intervals intersect and \(F\) when they are
disjoint. A contraction either leaves these payments unchanged or increases
them from \(F/2\) to \(F\), so the rule is core-selecting and SCCP.
Theorem~\ref{thm:three} establishes existence for \(n=3\). For \(n\geq4\), Theorem~\ref{thm:ccp-impossibility} rules out core-selecting CCP
rules and therefore also core-selecting SCCP rules.
\end{proof}

\section{One-Sided Domain with Endpoint Verification}\label{sec:onesided}

Theorem~\ref{thm:population-boundary} shows that arbitrary subinterval reports
on the  two-sided domain make core selection incompatible with CCP, and
hence with SCCP, once $n\geq4$. This section examines how the conclusion changes
when one endpoint is verified. We first consider common verified endpoints, then establish possibility for up to five customers and impossibility from six customers onward, yielding the exact population boundaries. The four-customer construction in
Theorem~\ref{thm:ccp-impossibility} relies on two deviations with opposing primary
directions and therefore does not apply to a uniformly one-sided contraction
domain.


Throughout this section, we consider the left-only domain. Each customer's right endpoint is verified, so every admissible contraction is of the form $I_i=[\ell_i,r_i]\longrightarrow J_i=[\ell_i',r_i]$, with $\ell_i\leq\ell_i'\leq r_i$.  The right-only domain is symmetric, that is, the reflection $[\ell_i,r_i]\longmapsto[-r_i,-\ell_i]$ transforms every right-only contraction into a left-only contraction. Therefore, it is sufficient to state and prove the results for the left-only domain.

\subsection{Common verified endpoints}\label{subsec:common-endpoints}

We first consider a structured domain in which all customers share the same verified right endpoint. This captures settings in which customers have a common certified deadline but may conceal how early they are available.

\begin{proposition}[Common verified right endpoint]
Suppose $r_i=R$ for every $i\in N$. Then, for every $n$, the equal-split rule 
$$x_i(I)=\frac{F}{n},\qquad i\in N,$$
is core-selecting and SCCP on the left-only domain.
\end{proposition}
\begin{proof}
Since $R\in I_i$ for every $i\in N$, all members of any nonempty coalition $S\subseteq N$ can be served on the single date $R$. Therefore, $\tau_I(S)=1$ and $c_I(S)=F$ for every nonempty $S\subseteq N$. The equal-split rule satisfies $\sum\limits_{i\in N}x_i(I)=c_I(N)$. Moreover, for every nonempty coalition  $S\subseteq N$, we have $\sum\limits_{i\in S}x_i(I)=\frac{|S|F}{n} \leq F =c_I(S)$. Thus, the rule is core-selecting.

Because the rule is independent of the reported intervals, every coalition's total payment remains unchanged under any admissible contraction. Hence, the rule is SCCP and therefore also CCP.
\end{proof}

\subsection{Population boundaries}

\subsubsection{Possibility for \texorpdfstring{$n\leq 5$}{n <= 5}}
\label{subsec:one-sided-positive}

We first construct a core-selecting SCCP rule for $n\leq 5$.

Consider the left-only domain. Since right endpoints are verified, order the customers, using a fixed public rule to break ties, so that $r_1\leq r_2\leq\cdots\leq r_n$.  Fix $I\in\mathcal{I}_n$ and let $m=\tau_I(N)$. The extremal-endpoint rule selects $m$ customers as follows:
\begin{enumerate}
    \item Select customer $1$.
    \item If $m\geq 2$, repeat the following step $m-2$ times: after selecting customer $k$, select the smallest-index customer $j$ satisfying $\ell_j>r_k$.
    \item After these selections, select, among all customers $j$ satisfying $\ell_j>r_k$, one with the largest left endpoint $\ell_j$, breaking ties by the smallest index.
    \item Every selected customer pays $F$, and every other customer pays $0$.
\end{enumerate}
If $m=1$, the procedure stops after the first step.

\begin{theorem}[One-sided possibility for $n\leq5$]
\label{thm:one-sided-positive}
For every $n\leq5$, the extremal-endpoint rule is core-selecting and
SCCP on the left-only domain. Therefore, it is also CCP.
The symmetric result holds on the right-only domain.
\end{theorem}
\begin{proof}
At every nonfinal step, the rule selects the eligible customer with the smallest verified right endpoint. This choice does not reduce the number of intervals that can subsequently be selected. Hence, a final customer always exists.

We first show that the rule is core-selecting. Fix $I\in\mathcal{I}_n$ and let $m=\tau_I(N)$. The rule selects $m$ customers whose intervals are pairwise disjoint. Therefore, $\sum\limits_{i\in N}x_i(I)=mF=c_I(N)$. For every $S\subseteq N$, the number of selected customers in $S$ is at most $\tau_I(S)$.  Hence, $\sum\limits_{i\in S}x_i(I)\leq F\tau_I(S)=c_I(S)$.  We next show that the rule is SCCP. Consider a left-only contraction from $I$ to $J$ and a coalition containing every contracting customer. Since $J_i\subseteq I_i$, every set of dates that serves $J$ also serves $I$. Therefore, $\tau_J(N)\geq\tau_I(N)$. Customer $1$ is selected at both profiles. Because every payment is either $0$ or $F$, it is enough to show that the number of newly selected contracting customers is at least the number of customers selected at $I$ but not at $J$.

If no customer selected at $I$ is omitted at $J$, the result is immediate. Suppose exactly one customer is omitted, and let $m=\tau_I(N)$. Compare the first $m-1$ selections at the two profiles. If they differ, consider their first difference. The preceding customer is the same at both profiles; let $a$ be selected next at $I$ and $c$ at $J$. Customer $a$ remains eligible at $J$, so the selection of $c$ implies $c<a$. Customer $c$ was not eligible at $I$, since otherwise the rule would have selected $c$ instead of $a$. Moreover, she cannot become eligible at a later step at $I$, because the right endpoints of subsequently selected customers are larger. Thus, customer $c$ contracted and is selected only at $J$.

Now suppose that the first $m-1$ selections are the same. The omitted customer is then the final customer selected at $I$. If $\tau_J(N)=m$, let $b$ and $d$ be the final customers selected at $I$ and $J$. If $d$ had not contracted, she would also have been eligible at $I$. Since $b$ was selected over $d$ at $I$, and $b$'s left endpoint can only increase while $d$'s is unchanged, $d$ cannot be selected over $b$ at $J$. Hence, $d$ contracted.  Now suppose that $\tau_J(N)>m$. The first $m-1$ selected customers are the same at both profiles, while the original final customer is no longer selected. Hence, all remaining customers selected at $J$ are newly selected. If none of them contracted, their intervals would be unchanged. These intervals would then be pairwise disjoint at $I$ and would begin after the verified right endpoint of the $(m-1)$st selected customer. Together with the first $m-1$ selected intervals, they would give more than $m$ pairwise-disjoint intervals at $I$, contradicting $\tau_I(N)=m$. Therefore, at least one newly selected customer must have contracted.
 Thus, one omitted customer is always offset by a newly selected contracting customer.

It remains to consider at least two omitted customers. Since $\tau_J(N)\geq\tau_I(N)$, at least two customers must also be newly selected at $J$. Together with customer $1$, who is selected at both profiles, the assumption $n\leq 5$ implies that $n=5$ and the selected customers are $\{1,a,b\}$ at $I$ and $\{1,c,d\}$  at $J$, where $a$ and $c$ are selected second and $b$ and $d$ are selected last.

Customer $a$ remains eligible after customer $1$ at $J$. Since the rule selects $c$, we have $c<a$. Customer $c$ was not eligible at $I$, and therefore $\ell_c(I)\leq r_1<\ell_c(J)$. Hence, customer $c$ contracted.

Since $c<a$, we have $r_c\leq r_a$. Moreover, $\ell_b(J)\geq\ell_b(I)>r_a\geq r_c$,  so customer $b$ remains eligible for the final selection at $J$. Because the rule selects $d$, $\ell_d(J)\geq\ell_b(J)>r_a$.  If $d$ had not contracted, she would also have been eligible for the final selection at $I$. Customer $b$ was selected over $d$ at $I$, and $b$'s left endpoint can only increase while $d$'s remains unchanged. Under the same tie-breaking rule, $d$ cannot be selected over $b$ at $J$, a contradiction. Thus, $d$ contracted.

Therefore, every customer who stops paying $F$ is offset by a newly paying contracting customer. No coalition containing all contracting customers can reduce its total payment, so the rule is SCCP. The right-only result follows by reflection.
\end{proof}

\subsubsection{Impossibility for \texorpdfstring{$n\geq6$}{n >= 6}}\label{subsec:one-sided-impossibility}

The common-endpoint result in Section~\ref{subsec:common-endpoints} applies for every population size. With different verified right endpoints, however, we now show that no core-selecting CCP rule exists from six customers onward. The construction uses a parameter $\rho$ to control the size of the contractions. Fix $0<\rho<\frac{1}{2}$ and consider
\begin{align*}
I_\rho^A
&=\left([0,1],[1,2],
\left[1+\frac{\rho}{2},3\right],[2+\rho,4],
\left[3+\frac{\rho}{2},5\right],[5+\rho,6]\right),\\
J_\rho
&=\left([0,1],[1+\rho,2],
[1+\rho,3],[2+\rho,4],
[3+\rho,5],[5+\rho,6]\right),\\
I_\rho^B
&=\left([0,1],[1+\rho,2],
[1+\rho,3],\left[2+\frac{\rho}{2},4\right],
[3+\rho,5],[5,6]\right).
\end{align*}

To obtain $J_\rho$ from $I_\rho^A$, customers $2$, $3$, and $5$ increase their left endpoints. To obtain $J_\rho$ from $I_\rho^B$, customers $4$ and $6$ increase their left endpoints. Therefore, $D(I_\rho^A,J_\rho)=\{2,3,5\}$, and $D(I_\rho^B,J_\rho)=\{4,6\}$.

It can be verified that the unique core allocations at $I_\rho^A$ and $I_\rho^B$ are  $x(I_\rho^A)=(F,0,F,0,F,F)$ and $x(I_\rho^B)=(F,F,0,F,0,F)$, respectively. At $J_\rho$, every core allocation satisfies
$$x_1(J_\rho)=x_6(J_\rho)=F,\qquad
x_2(J_\rho)+x_3(J_\rho)=F,\qquad
x_4(J_\rho)+x_5(J_\rho)=F.
$$

\begin{theorem}[One-sided CCP impossibility]
\label{thm:one-sided-ccp}
For every $n\geq 6$, no core-selecting CCP rule exists on the left-only domain. Therefore, no core-selecting SCCP rule exists on this domain. The symmetric conclusions hold on the right-only domain.
\end{theorem}
\begin{proof}
Let $n=6$ and $x$ be  a core-selecting CCP rule. For the contraction $I^A_\rho\to J_\rho$, CCP requires $x_2(J_\rho)+x_3(J_\rho)+x_5(J_\rho) \geq x_2(I^A_\rho)+x_3(I^A_\rho)+x_5(I^A_\rho)  =2F$.  Since $x_2(J_\rho)+x_3(J_\rho)=F$, this forces $x_5(J_\rho)=F$.

For the contraction $I^B_\rho\to J_\rho$, CCP requires $x_4(J_\rho)+x_6(J_\rho)
\geq x_4(I^B_\rho)+x_6(I^B_\rho) =2F$. Since $x_6(J_\rho)=F$, this forces $x_4(J_\rho)=F$. Therefore, $x_4(J_\rho)+x_5(J_\rho)=2F$, 
contradicting $x_4(J_\rho)+x_5(J_\rho)=F$.

For $n>6$, append the same fixed interval to all three profiles for every additional customer $h\in\{7,\ldots,n\}$: 
$$
I_{\rho,h}^A =I_{\rho,h}^B=J_{\rho,h} =[2h,2h+1].
$$
These additional intervals are pairwise disjoint and are disjoint from the intervals of customers $\{1,\ldots,6\}$. For every resulting profile $P$, $\tau_P(N)=\tau_P(\{1,\ldots,6\})+(n-6)$. Moreover, $\tau_P(N\setminus\{h\})=\tau_P(N)-1$ for every $h\in\{7,\ldots,n\}$. 
The total-payment condition and the core constraint for $N\setminus\{h\}$ imply $x_h(P)\geq F$, while the singleton core constraint gives $x_h(P)\leq F$. Thus, every additional customer pays $F$. Therefore, the payments of customers $\{1,\ldots,6\}$ satisfy the same core restrictions as in the original six-customer profiles. Since the additional customers do not change their reports, the same two CCP inequalities yield the contradiction. The right-only result follows by reflection.
\end{proof}

Combining the preceding results gives the exact population boundary on the one-sided domains.
\begin{theorem}[Exact population boundaries]
On either the left-only or the right-only  domain, a core-selecting SCCP rule exists if and only if $n\leq 5$. The same boundary holds for core-selecting CCP rules.
\end{theorem}

\begin{proof}
For $n\leq 5$, Theorem~\ref{thm:one-sided-positive} provides a core-selecting SCCP rule on the
left-only domain, and hence also a core-selecting CCP rule. For $n\geq 6$,
Theorem~\ref{thm:one-sided-ccp} rules out core-selecting CCP rules and hence also core-selecting
SCCP rules. Reflection gives the same conclusions on the right-only domain.
\end{proof}

\section{Robustness}
\label{sec:robustness}

\subsection{Relative concealment and robustness}

For a contraction from profile $I$ to profile $J$, define the fraction of
customer $i$'s original interval that is concealed by
\begin{equation}
q_i(I,J)
=1-\frac{|J_i|}{|I_i|}
=\frac{(\ell'_i-\ell_i)+(r_i-r'_i)}{r_i-\ell_i},
\label{eq:concealment-ratio}
\end{equation}
whenever $|I_i|>0$. Thus, $q_i(I,J)=0$ means that no flexibility is
concealed, while, for example, $q_i(I,J)=0.30$ means that 30 percent of the
original interval is removed. We restrict attention to contractions satisfying
\[
q_i(I,J)\leq \bar q
\qquad (i\in D(I,J)),
\]
where $\bar q\in[0,1]$. Singleton intervals need no concealment ratio because
they cannot be contracted strictly.

For the four-customer construction in Theorem~\ref{thm:ccp-impossibility},
the two contractions satisfy
\[
q_2(I_\rho^A,J_\rho)=\frac{\rho}{2+\rho},
\qquad
q_4(I_\rho^A,J_\rho)=\rho,
\]
and
\[
q_1(I_\rho^B,J_\rho)=\rho,
\qquad
q_3(I_\rho^B,J_\rho)=\frac{\rho}{2+\rho}.
\]
Hence the largest concealment ratio is $\rho$, which converges to zero as $\rho$ approaches zero. Therefore, for every $\bar q>0$, we can choose $\rho>0$
small enough that every contracting customer satisfies
$q_i(I,J)\leq\bar q$. The impossibility in
Theorem~\ref{thm:ccp-impossibility} thus persists even when customers conceal
an arbitrarily small fraction of their flexibility. The contractions remain
fully admissible on the two-sided domain: different customers may contract from different ends of their intervals.

For the six-customer construction, the nonzero concealment ratios are
\[
\begin{aligned}
q_2(I_\rho^A,J_\rho)=q_6(I_\rho^B,J_\rho)&=\rho,\\
q_3(I_\rho^A,J_\rho)=q_5(I_\rho^A,J_\rho)
=q_4(I_\rho^B,J_\rho)&=\frac{\rho}{4-\rho}.
\end{aligned}
\]
Their maximum is again $\rho$, so the one-sided impossibility also
persists for every positive concealment cap.

\subsection{Approximate requirements on the two-sided domain}
\label{sec:approx}

Let $\varepsilon,\delta\geq0$ and write
$x_T(P)=\sum_{i\in T}x_i(P)$ for the total payment of coalition $T$ at
profile $P$.

A rule is $\varepsilon$-core if, for every profile $P$,
\[
x_N(P)=F\tau_P(N)
\]
and
\[
x_S(P)\leq F\tau_P(S)+\varepsilon F
\qquad\text{for every }S\subseteq N.
\]
Thus, $\varepsilon=0$ gives the usual core requirement.

A rule is $\delta$-CCP if, for every admissible contraction $I\to J$,
\[
x_{D(I,J)}(J)
\geq
x_{D(I,J)}(I)-\delta F.
\]
Hence the contracting customers may reduce their combined payment by at most
$\delta F$. A rule is $\delta$-SCCP if the same inequality holds for every
coalition $C\supseteq D(I,J)$, with $C$ replacing $D(I,J)$. Therefore,
$\delta$-SCCP implies $\delta$-CCP.

We will repeatedly use one simple consequence of the $\varepsilon$-core
condition. For any $T\subseteq N$, applying the core inequality to the
complement $N\setminus T$ and using efficiency gives
\begin{equation}
\begin{aligned}
x_T(P)
&=x_N(P)-x_{N\setminus T}(P)\\
&\geq F\tau_P(N)-\bigl(F\tau_P(N\setminus T)+\varepsilon F\bigr)\\
&=\bigl[\tau_P(N)-\tau_P(N\setminus T)-\varepsilon\bigr]F.
\end{aligned}
\label{eq:approx-lower-bound}
\end{equation}
Equation~\eqref{eq:approx-lower-bound} converts the upper-bound form of the
$\varepsilon$-core constraints into a lower bound on the payment of a chosen
group $T$.

\begin{theorem}[Two-sided approximate trade-off]
\label{thm:approx}
Let $n\geq4$ and $\bar q>0$. Every $\varepsilon$-core,
$\delta$-CCP rule on the two-sided domain, even when attention is restricted to contractions satisfying
\[
q_i(I,J)\leq\bar q
\qquad(i\in D(I,J)),
\]
must satisfy
\[
4\varepsilon+2\delta
\geq
\left\lfloor\frac n4\right\rfloor.
\]
The same condition is necessary under $\delta$-SCCP.
\end{theorem}
\begin{proof}
Write $n=4k+r$, where $k=\lfloor n/4\rfloor$ and $0\leq r\leq3$.
Place $k$ copies of the four-customer construction in disjoint parts
of the time line and add $r$ isolated customers. These additional
customers make a small contraction only in the first deviation.
Choose all contractions small enough to satisfy the concealment cap.

Let $G_j$ contain the customers in position $j$ across the copies,
and let $R$ contain the additional customers. The contracting sets are
\[
D_A=G_2\cup G_4\cup R,\qquad
D_B=G_1\cup G_3,
\]
which partition $N$.

At $I^A$, deleting $G_2$ saves $k$ dates, while deleting $G_4\cup R$
saves $k+r$. At $I^B$, deleting either $G_1$ or $G_3$ saves $k$ dates.
Applying \eqref{eq:approx-lower-bound} separately to these groups
and adding gives
\[
x_{D_A}(I^A)\geq(2k+r-2\varepsilon)F,\qquad
x_{D_B}(I^B)\geq(2k-2\varepsilon)F.
\]
Since $J$ requires $3k+r$ dates, efficiency and the two
$\delta$-CCP inequalities imply
\[
\begin{aligned}
(3k+r)F
&=x_{D_A}(J)+x_{D_B}(J)\\
&\geq x_{D_A}(I^A)+x_{D_B}(I^B)-2\delta F\\
&\geq(4k+r-4\varepsilon-2\delta)F.
\end{aligned}
\]
Thus $4\varepsilon+2\delta\geq k$.
The same bound holds under $\delta$-SCCP, which implies $\delta$-CCP.
\end{proof}

For $n=4$, the condition becomes $2\varepsilon+\delta\geq\frac12$. More generally, an exact core ($\varepsilon=0$) requires $\delta\geq\frac12\left\lfloor\frac n4\right\rfloor$, whereas exact CCP ($\delta=0$) requires $\varepsilon\geq\frac14\left\lfloor\frac n4\right\rfloor$. These are necessary conditions, not sufficient ones.

\subsection{Approximate requirements on the one-sided domains}
\label{sec:approx-one-sided}

The six-customer obstruction gives the corresponding one-sided bound.

\begin{theorem}[One-sided approximate trade-off]
\label{thm:approx-one-sided}

Let $n\geq6$ and $\bar q>0$. Every $\varepsilon$-core rule satisfying
$\delta$-CCP for all left-only contractions with
$q_i(I,J)\leq\bar q$ for every $i\in D(I,J)$ must satisfy
\[
5\varepsilon+2\delta
\geq
\left\lfloor\frac n6\right\rfloor.
\]
The same condition holds on the right-only domain and under
$\delta$-SCCP.
\end{theorem}

\begin{proof}
Write $n=6k+r$, where $k=\lfloor n/6\rfloor$ and $0\leq r\leq5$.
Place $k$ copies of the six-customer construction in disjoint parts
of the time line and add $r$ isolated customers. These additional
customers make a small left-only contraction only in the first
deviation. Choose all contractions to satisfy the concealment cap.

Let $G_j$ contain the customers in position $j$ across the copies,
and let $R$ contain the additional customers. Then
\[
D_A=G_2\cup G_3\cup G_5\cup R,\qquad
D_B=G_4\cup G_6,
\]
so $D_A,D_B,G_1$ partition $N$.

At $I^A$, deleting $G_2\cup G_3$ saves $k$ dates, while deleting
$G_5\cup R$ saves $k+r$. At $I^B$, deleting either $G_4$ or $G_6$
saves $k$ dates. Applying \eqref{eq:approx-lower-bound} separately
and adding gives
\[
x_{D_A}(I^A)\geq(2k+r-2\varepsilon)F,\qquad
x_{D_B}(I^B)\geq(2k-2\varepsilon)F.
\]
Deleting $G_1$ at $J$ saves $k$ dates, so the same equation gives
$x_{G_1}(J)\geq(k-\varepsilon)F$.
Since $J$ requires $4k+r$ dates, efficiency and $\delta$-CCP imply
\[
\begin{aligned}
(4k+r)F
&=x_{D_A}(J)+x_{D_B}(J)+x_{G_1}(J)\\
&\geq x_{D_A}(I^A)+x_{D_B}(I^B)-2\delta F
      +(k-\varepsilon)F\\
&\geq(5k+r-5\varepsilon-2\delta)F.
\end{aligned}
\]
Hence $5\varepsilon+2\delta\geq k$.
Reflection gives the right-only result, and $\delta$-SCCP implies
$\delta$-CCP.
\end{proof}

For $n=6$, the bound becomes $5\varepsilon+2\delta\geq1$. More generally, setting $\varepsilon=0$ implies $\delta\geq
\frac12\left\lfloor\frac n6\right\rfloor$, and setting $\delta=0$ implies $\varepsilon\geq \frac15\left\lfloor\frac n6\right\rfloor$. These conditions are necessary, not sufficient.

\section{Discussion and Conclusion}\label{sec:conclusion}

\subsection{Pareto contraction-proofness}
We consider a weaker incentive requirement. A cost-sharing rule $x$
is Pareto contraction-proof if contracting customers cannot all weakly reduce
their payments, with at least one of them paying strictly less. Formally, for
every profile $I$, there is no admissible submitted profile $J$ such that
\[
x_i(J)\leq x_i(I)
\qquad\text{for every } i\in D(I,J),
\]
with strict inequality for at least one contracting customer. Note that CCP
implies Pareto contraction-proofness.
\subsubsection{Two-sided domain}

The four-customer construction in Theorem~\ref{thm:ccp-impossibility} also establishes impossibility under Pareto
contraction-proofness. For the contraction from $I_\rho^A$ to $J_\rho$,
customer 4's payment remains $F$, while customer 2's payment changes
from $F$ to $\alpha F$. Therefore, Pareto contraction-proofness  requires
$\alpha=1$. For the contraction from $I_\rho^B$ to $J_\rho$, customer 1's
payment remains $F$, while customer 3's payment changes from $F$ to
$(1-\alpha)F$. Therefore, Pareto contraction-proofness requires $\alpha=0$. This contradicts the requirement $\alpha=1$ obtained from the first contraction.

For $n>4$, the argument extends by adding isolated customers as in the
proof of Theorem~\ref{thm:ccp-impossibility}. Since the construction is unchanged, the impossibility
also persists under arbitrarily small relative concealment. Moreover,
SCCP implies Pareto contraction-proofness, so the positive results for
$n\leq3$ also apply. Thus, on the two-sided domain, a core-selecting
Pareto contraction-proof rule exists if and only if $n\leq3$.

\subsubsection{One-sided domain}

We show that the greedy rule is Pareto contraction-proof on the
left-only domain for every number of customers.

\begin{proposition}
\label{prop:one-sided-pareto}
For every $n$, the greedy rule $x^{\mathrm G}$ is core-selecting and
Pareto contraction-proof on the left-only domain. The reflected greedy rule $x^{\mathrm L}$ is core-selecting and
Pareto contraction-proof on the right-only domain.
\end{proposition}

\begin{proof}
Core selection follows from Theorem~\ref{thm:greedy-icp}.
Consider a left-only contraction from $I$ to $J$. Order customers by
their right endpoints, breaking ties by the smallest customer index.
Since right endpoints are verified, this order is unchanged.

If all payments remain unchanged, no customer benefits. Otherwise,
let $i$ be the first customer in this order whose payment changes.
All earlier selected customers, and hence all earlier service dates,
are the same at both profiles. If customer $i$ were selected at $I$,
none of those dates would belong to $I_i$, and hence none would
belong to $J_i$. Therefore, she would also be selected at $J$,
contrary to her payment changing. Thus, her payment increases
from zero to $F$.

Moreover, customer $i$ must have contracted her interval: if her
interval were unchanged, the same earlier service dates would give
the same selection decision. Hence, whenever payments change, at
least one contracting customer pays more. This proves Pareto
contraction-proofness. Reflection of the time axis gives the
right-only result.
\end{proof}

\subsection{Shapley value and nucleolus}

Our results also raise the question of whether standard cooperative solution concepts satisfy the requirements considered in this paper. Two standard candidates for allocating the service cost are the
Shapley value and the nucleolus. Consider the profile $I=([0,2],[1,4],[3,5])$, whose unique core allocation is $(F,0,F)$. For this profile, the Shapley value is $\phi(I)=\left(\frac{5F}{6},\frac{F} {3},\frac{5F}{6}\right)$.  Therefore, customers 1 and 2 pay a total of $7F/6$, although their stand-alone cost is $F$. Hence, the Shapley value is not core-selecting.

The nucleolus belongs to the core whenever the core is nonempty, so its allocation at $I$ is $\nu(I)=(F,0,F)$. Now suppose that customers 1 and 3 contract their intervals, giving $J=([0,1/2],[1,4],[7/2,5])$.  At $J$, the core consists of allocations of the form $(F,tF,(1-t)F)$, where $0\leq t\leq1$. Customers 2 and 3 are symmetric in the reported cost game, so the nucleolus assigns $\nu(J)=(F,F/2,F/2)$. Customer 1's payment remains unchanged, while customer 3's payment decreases from $F$ to $F/2$. Thus, the nucleolus fails both CCP and Pareto contraction-proofness on the two-sided domain, even with three customers.

\subsection{Conclusion}

This paper studies whether a cost-sharing rule can consistently select core allocations while preventing coalitions from reducing their joint payment by concealing service-time flexibility. On the two-sided domain, a core-selecting SCCP rule exists if and only if $n\leq3$. For $n\geq4$, no core-selecting rule satisfies even CCP. The four-customer contradiction requires only two contracting customers in each deviation and persists under arbitrarily small concealment. Customers are fully free to contract from the left, from the right, or from both sides; therefore the impossibility applies to the unrestricted two-sided domain.

Verifying one endpoint changes the exact boundary. If all customers share a common verified endpoint, the equal-split rule is core-selecting and SCCP for every population size. On the left-only or right-only domains, the extremal-endpoint rule is core-selecting and SCCP for $n\leq5$, whereas no core-selecting CCP rule exists for $n\geq6$. Uniform endpoint verification  raises the boundary from three to five, while a common verified endpoint removes the population restriction entirely. The four-customer impossibility also holds under Pareto contraction-proofness. On either one-sided domain, however, a core-selecting Pareto contraction-proof rule exists for every number of customers.

The obstruction also produces population-dependent necessary  trade-offs under approximate requirements. On the  two-sided domain, every $\varepsilon$-core, $\delta$-CCP rule with $n\geq4$ must satisfy
$4\varepsilon+2\delta\geq\left\lfloor\frac{n}{4}\right\rfloor$. On either one-sided domain, every such rule with $n\geq6$ must satisfy $5\varepsilon+2\delta\geq\left\lfloor\frac{n}{6}\right\rfloor$. Both conditions remain necessary under $\delta$-SCCP and under arbitrarily small relative concealment.

\bibliographystyle{plainnat}
\bibliography{references}

@article{agatz2011,
  title={Time slot management in attended home delivery},
  author={Agatz, Niels and Campbell, Ann and Fleischmann, Moritz and Savelsbergh, Martin},
  journal={Transportation Science},
  volume={45},
  number={3},
  pages={435--449},
  year={2011},
  publisher={Informs}
}

@article{baheltrudeau2022,
  title={Minimum coloring problems with weakly perfect graphs},
  author={Bahel, Eric and Trudeau, Christian},
  journal={Review of Economic Design},
  volume={26},
  number={2},
  pages={211--231},
  year={2022},
  publisher={Springer}
}

@article{bietenhader2006,
  title={Core stability of minimum coloring games},
  author={Bietenhader, Thomas and Okamoto, Yoshio},
  journal={Mathematics of Operations Research},
  volume={31},
  number={2},
  pages={418--431},
  year={2006},
  publisher={INFORMS}
}

@article{deligkas2024,
  title={Truthful interval covering},
  author={Deligkas, Argyrios and Filos-Ratsikas, Aris and Voudouris, Alexandros A},
  journal={Autonomous Agents and Multi-Agent Systems},
  volume={38},
  number={2},
  pages={41},
  year={2024},
  publisher={Springer}
}

@article{deng1999,
  title={Algorithmic aspects of the core of combinatorial optimization games},
  author={Deng, Xiaotie and Ibaraki, Toshihide and Nagamochi, Hiroshi},
  journal={Mathematics of Operations Research},
  volume={24},
  number={3},
  pages={751--766},
  year={1999},
  publisher={INFORMS}
}

@article{frisk2010,
  title={Cost allocation in collaborative forest transportation},
  author={Frisk, Mikael and G{\"o}the-Lundgren, Maud and J{\"o}rnsten, Kurt and R{\"o}nnqvist, Mikael},
  journal={European journal of operational research},
  volume={205},
  number={2},
  pages={448--458},
  year={2010},
  publisher={Elsevier}
}

@article{greenlaffont1986,
  title={Partially verifiable information and mechanism design},
  author={Green, Jerry R and Laffont, Jean-Jacques},
  journal={The Review of Economic Studies},
  volume={53},
  number={3},
  pages={447--456},
  year={1986},
  publisher={Wiley-Blackwell}
}

@article{guajardoronnqvist2016,
  title={A review on cost allocation methods in collaborative transportation},
  author={Guajardo, Mario and R{\"o}nnqvist, Mikael},
  journal={International transactions in operational research},
  volume={23},
  number={3},
  pages={371--392},
  year={2016},
  publisher={Wiley Online Library}
}

@article{hamers2014,
  title={Monotonic stable solutions for minimum coloring games},
  author={Hamers, Herbert and Miquel, Silvia and Norde, Henk},
  journal={Mathematical Programming},
  volume={145},
  number={1},
  pages={509--529},
  year={2014},
  publisher={Springer}
}

@article{kartiktercieux2012,
  title={Implementation with evidence},
  author={Kartik, Navin and Tercieux, Olivier},
  journal={Theoretical Economics},
  volume={7},
  number={2},
  pages={323--355},
  year={2012},
  publisher={Wiley Online Library}
}

@article{klein2019,
  title={Differentiated time slot pricing under routing considerations in attended home delivery},
  author={Klein, Robert and Neugebauer, Michael and Ratkovitch, Dimitri and Steinhardt, Claudius},
  journal={Transportation Science},
  volume={53},
  number={1},
  pages={236--255},
  year={2019},
  publisher={INFORMS}
}

@article{krahmerstrausz2025,
  title={Unidirectional incentive compatibility},
  author={Kr{\"a}hmer, Daniel and Strausz, Roland},
  journal={Journal of Economic Theory},
  volume={228},
  pages={106051},
  year={2025},
  publisher={Elsevier}
}

@article{moulin1999,
  title={Incremental cost sharing: Characterization by coalition strategy-proofness},
  author={Moulin, Herv{\'e}},
  journal={Social Choice and Welfare},
  volume={16},
  number={2},
  pages={279--320},
  year={1999},
  publisher={Springer}
}

@article{moulinshenker2001,
  title={Strategyproof sharing of submodular costs: budget balance versus efficiency},
  author={Moulin, Herv{\'e} and Shenker, Scott},
  journal={Economic Theory},
  volume={18},
  number={3},
  pages={511--533},
  year={2001},
  publisher={Springer}
}

@article{munich2024,
  title={Schedule situations and their cooperative game theoretic representations},
  author={Munich, L{\'e}a},
  journal={European Journal of Operational Research},
  volume={316},
  number={2},
  pages={767--778},
  year={2024},
  publisher={Elsevier}
}

@article{okamoto2008,
  title={Fair cost allocations under conflicts—a game-theoretic point of view—},
  author={Okamoto, Yoshio},
  journal={Discrete Optimization},
  volume={5},
  number={1},
  pages={1--18},
  year={2008},
  publisher={Elsevier}
}

@article{singhwittman2001,
  title={Implementation with partial verification},
  author={Singh, Nirvikar and Wittman, Donald},
  journal={Review of Economic Design},
  volume={6},
  number={1},
  pages={63--84},
  year={2001},
  publisher={Springer}
}

@article{solomon1987,
  title={Algorithms for the vehicle routing and scheduling problems with time window constraints},
  author={Solomon, Marius M},
  journal={Operations research},
  volume={35},
  number={2},
  pages={254--265},
  year={1987},
  publisher={Informs}
}

@article{solomondesrosiers1988,
  title={Survey paper—{T}ime window constrained routing and scheduling problems},
  author={Solomon, Marius M and Desrosiers, Jacques},
  journal={Transportation science},
  volume={22},
  number={1},
  pages={1--13},
  year={1988},
  publisher={INFORMS}
}

@article{tamir2023,
  title={Cost-sharing games in real-time scheduling systems},
  author={Tamir, Tami},
  journal={International Journal of Game Theory},
  volume={52},
  number={1},
  pages={273--301},
  year={2023},
  publisher={Springer}
}

@article{ulmer2024,
  title={Optimal service time windows},
  author={Ulmer, Marlin W and Goodson, Justin C and Thomas, Barrett W},
  journal={Transportation Science},
  volume={58},
  number={2},
  pages={394--411},
  year={2024},
  publisher={Informs}
}

@article{vinsensius2020,
  title={Dynamic incentive mechanism for delivery slot management in e-commerce attended home delivery},
  author={Vinsensius, Albert and Wang, Yuan and Chew, Ek Peng and Lee, Loo Hay},
  journal={Transportation Science},
  volume={54},
  number={3},
  pages={567--587},
  year={2020},
  publisher={INFORMS}
}

@article{vijayalakshmi2026,
  title={Interval Scheduling Games with Color-Based Concurrent Jobs},
  author={Ravindran Vijayalakshmi, Vipin and Schr{\"o}der, Marc and Tamir, Tami},
  journal={ACM Transactions on Economics and Computation},
  year={2026},
  publisher={ACM New York, NY}
}

@article{visser2024,
  title={Managing concurrent interactions in online time slot booking systems for attended home delivery},
  author={Visser, Thomas R and Agatz, Niels and Spliet, Remy},
  journal={Transportation Science},
  volume={58},
  number={5},
  pages={1056--1075},
  year={2024},
  publisher={INFORMS}
}

\end{document}